\documentclass[11pt,a4paper]{article}
\def\Label#1{}
\usepackage[utf8]{inputenc}

\pdfoutput=1 

\usepackage[expansion=true]{microtype}
\usepackage[plainpages=false,hidelinks]{hyperref}

\let\rho=\varrho

\def\pbarsq{{\bar p_2^{\,2}}}

\usepackage{ifthen}

\def\defcommand#1:#2?{\global\let\myref\myundefined%
	\ifthenelse{\equal{#1}{e}}{\global\let\myref\eref}{}%
	\ifthenelse{\equal{#1}{p}}{\global\let\myref\pref}{}%
	\ifthenelse{\equal{#1}{l}}{\global\let\myref\lref}{}%
	\ifthenelse{\equal{#1}{r}}{\global\let\myref\rref}{}%
	\ifthenelse{\equal{#1}{t}}{\global\let\myref\tref}{}%
	\ifthenelse{\equal{#1}{f}}{\global\let\myref\fref}{}%
	\ifthenelse{\equal{#1}{s}}{\global\let\myref\sref}{}%
	\ifthenelse{\equal{#1}{d}}{\global\let\myref\dref}{}%
}
\let\ssref=\ref
\def\fref#1{Figure~\ssref{#1}}

\def\cref#1{Condition~\ssref{#1}}
\def\Cref#1{Corollary~\ssref{#1}}
\def\eref#1{(\ssref{#1})}
\def\sref#1{\textsection\ssref{#1}}
\def\lref#1{Lemma~\ssref{#1}}
\def\rref#1{Remark~\ssref{#1}}
\def\tref#1{Theorem~\ssref{#1}}
\def\dref#1{Definition~\ssref{#1}}
\def\pref#1{Proposition~\ssref{#1}}

\def\dref#1{Definition~\ssref{#1}}
\def\myundefined#1{
{what is #1?}%
}
\def\ref#1{\defcommand#1?%
	\myref{#1}}
\let\myref\relax

\usepackage{amsmath}
\usepackage{amsfonts}
\usepackage{graphicx}
\usepackage{times}
\usepackage{amsthm}
\usepackage{ amssymb }
\usepackage{color}
\usepackage{mhequ}
\usepackage{dsfont}
\usepackage[margin=2.6cm]{geometry}
\usepackage{color}
\usepackage{url}
\usepackage{lipsum}
\usepackage{subfig}
\usepackage{mhequ}

\def \mEE {\mathbb E}

\usepackage[graphics, active, tightpage]{}

\def\thecomma{\ifx,\thenewxt \else\ifx;\thenext \else\ifx.\thenext
	\else\ifx!\thenext \else\ifx:\thenext\else\ifx)\thenext \else \
	\fi\fi\fi\fi\fi\fi}
\def\condblank{\futurelet\thenext\thecomma}
\def\ie{{\it i.e.,}\condblank}

\numberwithin{equation}{section}

\newtheorem{theorem}{Theorem}[section]
\newtheorem{lemma}[theorem]{Lemma}
\newtheorem{proposition}[theorem]{Proposition}
\newtheorem{definition}[theorem]{Definition}

\theoremstyle{definition} 
\newtheorem{remark}[theorem]{Remark}

\usepackage{cite}

\usepackage{stmaryrd}

\newcommand{\dd}{\mathrm{d}}

\newcommand{\dt}{\,\dd t}

\def\argcdot{{\,\cdot\,}}

\let\kappa=\varkappa
\let\phi=\varphi

\newcommand{\ind}{\mathbf{1}}

\let\epsilon=\varepsilon

\def\torus{{\mathbb T}}
\def\real{{\mathbb R}}
\def\integer{{\mathbb Z}}

\def\argcdot{{\,\cdot\,}}

\begin{document}

\title{On the relaxation rate of short chains of rotors interacting with Langevin thermostats}
\author{N. Cuneo\thanks{Department of Mathematics and Statistics, McGill University, 805 Sherbrooke St.~West, Montreal, QC, H3A 0B9, Canada} ~ and ~C. Poquet\thanks{Univ Lyon, Universit\'e Claude Bernard Lyon 1, CNRS UMR 5208, Institut Camille Jordan, 43 blvd. du 11 novembre 1918, F-69622 Villeurbanne cedex, France}}

\date{} 

\maketitle
\vspace{-0.5cm}
\thispagestyle{empty}
\begin{abstract}
In this short note, we consider a system of two rotors, one of which interacts with a Langevin heat bath. We show that the system relaxes to its invariant measure (steady state) no faster than a stretched exponential  $\exp(-c t^{1/2})$. This indicates that the exponent $1/2$ obtained earlier by the present authors and J.-P. Eckmann for short chains of rotors is optimal.
\end{abstract}

\section{Introduction and main result}

When Hamiltonian chains interact with stochastic heat baths, the rate of relaxation to the steady state (if there is one) is known to depend on the model.  For some chains of oscillators with strong interactions, exponential convergence has been proved in \cite{eckmann_nonequilibrium_1999,eckmann_hairer_2000,reybellet_exponential_2002,carmona_2007}. On the other hand, chains of oscillators where the pinning dominates the interactions can exhibit strictly subgeometric rates, as proved in \cite{hairer_slow_2009}. In fact, the convergence rate can be very sensitive to the parameters of the model, as nicely illustrated in \cite{hairer_how_2009} for a system of two oscillators with two heat baths, one of which has ``infinite temperature''. Other energy-exchange models where subgeometric convergence rates have been observed include \cite{li2016polynomial, li2016polynomial2, MR3217533, MR3071443}. 

It has been shown by the present authors and J.-P. Eckmann
that chains of three and four rotors, interacting with Langevin heat baths at both ends, relax to their invariant measure at least as fast as a stretched exponential $\exp(-c t^{1/2})$ \cite{CEP_nonequilibrium_2014,four_rotors_2015}. We address in this note the question of whether the exponent $1/2$ is optimal\footnote{For chains of length $n$, it is conjectured in \cite[Remark 5.3]{four_rotors_2015} that the exponent is $1/({2\lceil{n/2}\rceil-2})$, which is indeed $1/2$ when $n=3,4$. This conjecture is supported by \cite{CEW2017}, where the rate of energy dissipation in deterministic chains of rotors of arbitrary lengths is studied.}. We consider a chain of two rotors interacting with only one heat bath (see \fref{f:system2rotors}). This simple system can be seen as ``one half'' of the chains of length three and four considered in \cite{CEP_nonequilibrium_2014,four_rotors_2015}, and
we show, using techniques introduced in
\cite{hairer_how_2009}, that the convergence indeed happens no faster than the stretched exponential mentioned above (up to the value of the constant $c$).

\begin{figure}[ht]
\centering
\includegraphics[width=1.6in]{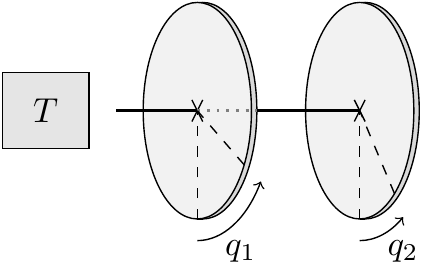}~~
\caption{The model.}
\label{f:system2rotors}
\end{figure}

The model is as follows. Each rotor has a position (angle) $q_i\in \torus =
\real/2\pi
\integer$ and a momentum $p_i \in \real$, $i=1,2$. The phase space is thus $\Omega = \torus^2 \times \real^2$, and we write $x = (q,p) = (q_1, q_2, p_1, p_2)$. 
The first rotor is coupled to a Langevin-type heat bath at temperature $T>0$, with a coupling constant $\gamma > 0$. The second rotor interacts with the first one through a smooth potential $W(q_2-q_1)$, and is not coupled to any heat bath. We apply no other external force to the system, and we take no pinning (on-site) potential. The Hamiltonian is
\begin{equ}
H(q, p) = {\frac {p_1^2 + p_2^2} 2} + W(q_2-q_1)~,
\end{equ}
and we consider the stochastic differential equation
\begin{equs}[eq:sdelowerbound]
	\dd q_i(t)&=p_i(t)\dt~, 
	\qquad\qquad\qquad\qquad\qquad\qquad\qquad 
	 i=1,2~,\\
	\dd p_1(t)&=w\big(q_2(t)-q_1(t)\big)\dt-\gamma
	p_1(t)\dt +\sqrt{2\gamma T}\,\dd B_t~,\\
	\dd p_2(t)&=-w\big(q_2(t)-q_1(t)\big)\dt ~,
\end{equs}
where we have introduced the derivative $w$ of $W$, and where $B_t$ is a standard Wiener process. Without loss of generality, we choose the additive constant in $W$ such that $\int_0^{2\pi} W(s) \dd s = 0$. Since the coefficients of \eref{eq:sdelowerbound} are globally Lipschitz, the process is non-explosive. We denote by $P^t(x, \argcdot)$ the transition probabilities, by $\mEE_x$ the expectation with respect to the process started at $x\in \Omega$, and we introduce the generator
\begin{equs}
L = p_1 \partial_{q_1} + p_2 \partial_{q_2}
+w(q_2-q_1) (\partial_{p_1} - \partial_{p_2}) - \gamma p_1 \partial_{p_1} + \gamma T \partial^2_{p_1}~.
\end{equs}

The physical picture is essentially the same as with three and four rotors (see \cite{CEP_nonequilibrium_2014,four_rotors_2015}): the second rotor decouples when its energy is large, and the crux is to obtain an effective dynamics in this regime.
Since there is only one temperature at hand, the invariant measure $\pi$ is simply the Gibbs measure, \ie
\begin{equs}
	\dd \pi (q, p) = \frac 1 Z e^{-\frac H T} \dd p \dd q~,
\end{equs}
where $Z$ is a normalization constant.

By making minor adaptations to the proof in \cite{CEP_nonequilibrium_2014} (actually, by simply omitting the third rotor), one obtains that the invariant measure $\pi$ is unique, and that there exist a constant $c>0$ and a function $h:\Omega \to \real^+$ such that
\begin{equ}[eq:upperbound]
\|P^{t}(x, \argcdot) - \pi \|_{\mathrm {TV}} \leq h(x)e^{-c \sqrt{t}}
\end{equ}
for all $x\in \Omega$ and $t\geq 0$. Here, $\|\argcdot\|_{\mathrm {TV}}$ denotes the total variation norm, but the result also holds for some stronger norms. For these stronger norms, the constant $c$ and the function $h(x)$ change, but the power of $t$ in the exponential remains $1/2$ (see \cite[Theorem 1.3]{CEP_nonequilibrium_2014}).

We now formulate the main result, which gives a {\em lower} bound on $\|P^{t}(x, \argcdot) - \pi \|_{\mathrm {TV}}$.
\begin{theorem}\label{thm:mainthmlowerbound}
There exist a constant $c_* > 0$ and a function $h_*: \Omega \to \real^+$ such that for each initial condition $x \in \Omega$, there is a sequence $(t_n)_{n\geq 0}$ increasing to infinity such that 
\begin{equ}
\|P^{t_n}(x, \argcdot) - \pi \|_{\mathrm {TV}} \geq h_*(x)e^{-c_* \sqrt{t_n}}~.
\end{equ}
\end{theorem}
\begin{proof}

In \pref{prop:constructtestfunction} below, we will construct a test function $F: \Omega \to [1, \infty)$ such that $\pi(F^{1-\varepsilon}) = \infty~$ for some $\varepsilon \in (0,1)$, and such that for some $A>0$ and all $x\in \Omega$, 
\begin{equs}[eq:borneEFxt]
\mEE_x F(x_t) \leq  F(x)  e^{\sqrt {2At}}~.
\end{equs}
The desired result then follows from \cite[Theorem 3.6 and Corollary 3.7]{hairer_how_2009}. For completeness, we give here an explicit adaptation of the proof to the present case.

We fix $x \in \Omega$ and write $\nu_t = P^t(x, \argcdot)$.
The result follows from comparing an upper bound on the tail of $\nu_t$ with a lower bound on the tail of $\pi$.\begin{itemize}
\item By \eref{eq:borneEFxt} and Markov's inequality,  we have for all $w > 0$ the upper bound
\begin{equ}[eq:ineqVW]
	\nu_t(F > w) \leq \frac {F(x)e^{\sqrt {2At}}}{w}~.
\end{equ}
\item Since $
(1-\varepsilon) \int_1^\infty \pi(F > w) w^{-\varepsilon} \dd w = \pi(F^{1-\varepsilon}) = \infty$, there is a sequence $(w_n)_{n\geq 0}$ increasing to infinity such that $\pi(F > w_n) w_n^{-\varepsilon} \geq w_n^{-1-\varepsilon/2}$. As a consequence, we have for each $n \geq 0$ the inequality
\begin{equs}[eq:ineqVwn]
\pi(F> w_n) \geq \frac 1 {w_n^{1-\varepsilon/2}}	~.
\end{equs}
\end{itemize}
By \eref{eq:ineqVW}, \eref{eq:ineqVwn} and the definition of the total variation norm, we have for all $n$ that
\begin{equs}
\|\nu_t - \pi \|_{\mathrm {TV}} \geq \pi(F > w_n) - \nu_t(F > w_n) \geq \frac 1 {w_n^{1-\varepsilon/2}} - \frac {F(x)e^{\sqrt {2At}}}{w_n}~.
\end{equs}
Picking now $t_n$ such that $F(x)e^{\sqrt {2At_n}} = \frac 12 w_n^{\varepsilon/2}$, we obtain 
\begin{equs}
\|\nu_{t_n} - \pi \|_{\mathrm {TV}} \geq \frac 1 {2 w_n^{1-\varepsilon/2}}  = h_*(x)e^{-c_*\sqrt {t_n}}~,
\end{equs}
with $c_* = (\frac 2 \varepsilon -1) \sqrt{2A}$ and $h_*(x) = \frac 12 (2F(x))^{1-\frac 2 \varepsilon}$. This completes the proof.
\end{proof}

\section{Construction of the test function}

We now construct a function $F$ that has the properties needed in the proof of \tref{thm:mainthmlowerbound}. This function $F$ will grow fast enough along the $p_2$-axis so that $\pi(F^{1-\varepsilon}) = \infty$ for all small enough $\varepsilon$. Moreover, $F$ will satisfy a relation of the kind $L F \lesssim F / \log F$, which implies \eref{eq:borneEFxt} as we will show.
The test function is basically built by averaging the rapid oscillations of the second rotor (see also \cite{cooke2011geometric} for a related approach in a different setup). 

We start by approximating the dynamics of $p_2$ by an ``averaged'' variable $\bar p_2$ in the regime where $p_2$ is very large.
In this regime, $p_2$ undergoes small and fast oscillations, and as in \cite{CEP_nonequilibrium_2014} and  §2 of \cite{four_rotors_2015}, we will remove step by step the oscillatory terms from its dynamics by adding well-chosen counter-terms to the variable $p_2$.
To obtain the desired result, we will need some control on the dynamics also when 
$p_1$ scales linearly with $p_2$ (see \pref{prop:constructtestfunction}), and thus we cannot simply use an expansion in negative powers of $p_2$ as in \cite{CEP_nonequilibrium_2014,four_rotors_2015}. We instead consider negative powers of $p_2-p_1$, with the following more refined notion of {\em order}  (see also \rref{rem:lowerbound2seulement}).
\begin{definition}
	For any continuous function $f :  \torus^2 \to \real$ and any choice of integers $k, \ell \geq 0$ and $m\in \mathbb Z$, we say that
\begin{equs}
\frac {f(q) p_1^k p_2^m}{(p_2-p_1)^\ell }	
\end{equs}
is {\em of order} $k+m-\ell $. We denote by $\mathcal R(j)$ a generic remainder of order at most $j$, \ie a finite sum of terms of order up to $j$.
\end{definition}

The usual rules apply, in particular $\mathcal R(j) + \mathcal R(k) = \mathcal R(\max(j,k))$, and $\mathcal R(j) \mathcal R(k) = \mathcal R(j+k)$.
The aim now is to introduce a new variable $\bar p_2 = p_2 + \mathcal R(-1)$, which is defined when $p_2 \neq p_1$, and which satisfies
\begin{equs}[eq:aimpbarsqtildelower]
\dd \bar p_2 = \mathcal R(-3) \dd t +  \mathcal R(-2) \dd B_t	~.
\end{equs}
The remainders $\mathcal R(-3)$ and $\mathcal R(-2)$ above need not be computed explicitly. We start with
\begin{equs}[eq:firststep]
\dd p_2 = -w(q_2-q_1) \dt ~.
\end{equs}
We then introduce a first correction
\begin{equs}[eq:firstcorrection]
p_2^{(1)} = p_2 + \frac{W(q_2-q_1)}{p_2-p_1}	~,
\end{equs}
and obtain by Itô's formula
\begin{equs}[eq:resfirststep]
\dd p_2^{(1)} =  \frac {W(q_2-q_1)} {(p_2-p_1)^2}(2w(q_2-q_1) - \gamma p_1) \dt + \mathcal R(-2) \dd B_t+  \mathcal R(-3)\dd t~.
\end{equs}
The counter-term in \eref{eq:firstcorrection} was chosen precisely so as to cancel the right-hand side of \eref{eq:firststep}, up to the higher order terms appearing in \eref{eq:resfirststep} (see also §3.3 of \cite{CEP_nonequilibrium_2014} for more explanations).

Since $\int_0^{2\pi} W(s) \dd s = 0 $, there exists an indefinite integral $W^{[1]}$ of $W$ on $\torus$, which we choose so that $\int_0^{2\pi} W^{[1]}(s) \dd s = 0 $. In turn, we introduce an indefinite integral $W^{[2]}$ of $W^{[1]}$. By construction, we have $\big(W^{[1]}\big)' = W$ and $\big(W^{[2]}\big)' = W^{[1]}$.

We then set
\begin{equs}
 p_2^{(2)}= p_2^{(1)} + \frac{\gamma p_1 W^{[1]}(q_2-q_1) - (W(q_2-q_1))^2}{(p_2-p_1)^3}~,
\end{equs}
and obtain
\begin{equs}
\dd p_2^{(2)} = -  \frac { \gamma^2 p_1 W^{[1]}(q_2-q_1)  } {(p_2-p_1)^3}\dt  - \frac {3 \gamma^2  p_1^2  W^{[1]}(q_2-q_1) } {(p_2-p_1)^4} \dt + \mathcal R(-2) \dd B_t+  \mathcal R(-3)\dd t~.
\end{equs}

We finally obtain \eref{eq:aimpbarsqtildelower} by letting
\begin{equs}
\bar p_2= p_2^{(2)} + \frac { \gamma^2 p_1 W^{[2]}(q_2-q_1)  } {(p_2-p_1)^4} + \frac {3 \gamma^2  p_1^2  W^{[2]}(q_2-q_1) } {(p_2-p_1)^5}~.
\end{equs}

In order to construct the test function $F$, we now introduce some positive parameters $\beta_-, \beta_+$ and $\delta$ satisfying
\begin{equs}[eq:condbetaspm]
\beta_- < \frac 1 T < \beta_+ < \left(1+\frac 1{(1+2\delta)^2}\right)\beta_-	~,
\end{equs}
and consider the partition of $\Omega$ (see \fref{fig:partitionlowerbound}) given by
\begin{itemize}
\item $\Omega_0 = \{x\in \Omega: p_1^2 + p_2^2 < 1\}$~,
\item $\Omega_1 = \{x\in \Omega: |p_2| \leq (1+\delta)|p_1|\}\setminus \Omega_0$~,
\item $\Omega_2 = \{x\in \Omega: (1+\delta)|p_1| < |p_2| \leq  (1+2\delta)|p_1|\}\setminus \Omega_0$~,
\item $\Omega_3 = \{x\in \Omega: |p_2| > (1+2\delta)|p_1|\}\setminus \Omega_0$~.
\end{itemize}

\begin{figure}[ht]
\centering
\includegraphics[width=2.5in, trim = 0 0mm 0mm 0, clip]{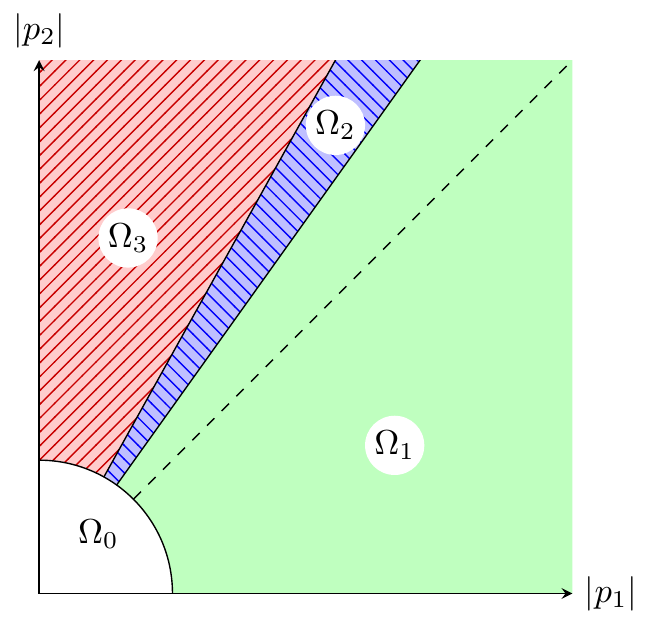}
\caption{Partition of $\Omega$ (in momentum space).}\label{fig:partitionlowerbound}
\end{figure}

We immediately have
\begin{lemma} There are constants $C_1$ and $C_2$ such that on the set $\Omega_2 \cup \Omega_3$, we have
the two inequalities
\begin{equ}\label{eq:lomegasfirstlower}
|\pbarsq - p_2^2| < C_1~,
\end{equ}
\begin{equ}\label{eq:lomegasscndlower}
L e^{\frac {\beta_+ }2  \pbarsq}  \leq C_2 p_2^{-2}e^{\frac {\beta_+ } 2  \pbarsq}~.
\end{equ}

\end{lemma}
\begin{proof} Observe that for all  $k, \ell \geq 0$ and $m\in \mathbb Z$, there is a constant $C$ such that on the set $\Omega_2 \cup \Omega_3 = \{x\in \Omega: |p_2| > (1+\delta)|p_1|,~ p_1^2 + p_2^2 \geq 1\}$, we have
\begin{equs}
\left|\frac {p_1^kp_2^m }{(p_2-p_1)^\ell }	\right| \leq C |p_2|^{k+m-\ell}~.
\end{equs}
This implies that any remainder $\mathcal R(j)$ is bounded in absolute value by some constant times $|p_2|^j$ on $\Omega_2 \cup \Omega_3$.
In particular, since $\pbarsq = (p_2 + \mathcal R(-1))^2 = p_2^2 + \mathcal R(0)$, we obtain that \eref{eq:lomegasfirstlower} holds on $\Omega_2 \cup \Omega_3$ for some
appropriate $C_1$.

In order to prove \eref{eq:lomegasscndlower}, we write $f(s) = e^{\frac {\beta_+} 2  s^2}$ and obtain by Itô's formula
\begin{equs}
\dd \big(e^{\frac {\beta_+} 2  \pbarsq}\big) & = \dd f(\bar p_2)
= f'(\bar p_2) ( \mathcal R(-3)\dt + \mathcal R(-2)\dd B_t) + \frac 12 f''(\bar p_2)\mathcal R(-4)  \dt~.
\end{equs}
We thus find, since $f'(\bar p_2) = \mathcal R(1)e^{\frac {\beta_+} 2  \pbarsq}$ and $f''(\bar p_2) = \mathcal R(2)e^{\frac {\beta_+} 2  \pbarsq}$, that
\begin{equs}
	L e^{\frac {\beta_+} 2  \pbarsq} & =  f'(\bar p_2) \mathcal R(-3)  + \frac 12 f''(\bar p_2)\mathcal R(-4)    = \mathcal R(-2)e^{\frac {\beta_+} 2  \pbarsq}~.
\end{equs}
Now, on the set $\Omega_2 \cup \Omega_3$, the $\mathcal R(-2)$ above is bounded by $C_2 p_2^{-2}$ for some $C_2>0$, and thus \eref{eq:lomegasscndlower} holds.
\end{proof}

We next introduce a smooth cutoff function $\rho:\real^2 \to [0,1]$ such that $\rho(p) =1$ on $\Omega_3$ and $\rho(p) = 0$ on $\Omega_1$, with some transition on $\Omega_2$. More precisely, let $\chi:[0, \infty] \to \real$ be a smooth function such that
 $\chi(s) = 1$ when $s \geq 1+2\delta$, and $\chi(s) = 0$ when $s \leq 1+\delta$. On $\Omega \setminus \Omega_0$, we let
\begin{equs}
\rho(p) = \chi \left(\left|\frac {p_2}{p_1}\right|\right)~,
\end{equs}
and we freely choose $\rho$ on $\Omega_0$ so that it is smooth on all of $\Omega$.

We now define the function $F: \Omega \to [1, \infty)$ by
\begin{equs}[eq:deffunctionf]
F(x) = 1 + e^{\beta_- H(x)} + \rho(p)e^{\frac{\beta_+}2 \pbarsq}	~,
\end{equs}
for some $\beta_-, \beta_+, \delta$ satisfying \eref{eq:condbetaspm}. Observe that while $F$ resembles the Lyapunov function $V$ of \cite{CEP_nonequilibrium_2014}, it grows much faster along the $p_2$-axis.

\begin{proposition}\label{prop:constructtestfunction}
Let $F$ be as defined in \eref{eq:deffunctionf}. Then, $\pi(F^{1-\varepsilon})=\infty$ for small enough $\varepsilon$, and \eref{eq:borneEFxt} holds for large enough $A$.
\end{proposition}
\begin{proof}
Let $\Gamma  = \{x: |p_1| \leq 1\} \cap \Omega_3$. Using \eref{eq:lomegasfirstlower} and the definition of $H$, we find
\begin{equs}[eq:nonintegrabilitylower]
	\pi(F^{1-\varepsilon}) &\geq \int_\Gamma   \exp\left({\frac{\beta_+(1-\varepsilon) \pbarsq}2}\right)	\frac{\exp\left({- \frac H T}\right)}{Z}\dd p \dd q \\
& \geq c_0 \int_\Gamma  \exp\left({\frac{\beta_+(1-\varepsilon) p_2^2}2}- \frac {p_2^2} {2T}\right)\dd p \dd q
\end{equs}
for some $c_0 > 0$.
Provided that we pick $\varepsilon$ small enough so that $\frac 1T < \beta_+ (1-\varepsilon)$, which is possible by \eref{eq:condbetaspm}, the last integral in \eref{eq:nonintegrabilitylower} is infinite, and thus $\pi(F^{1-\varepsilon}) = \infty$.

We now prove the second assertion. As in \cite{CEP_nonequilibrium_2014}, we introduce the concave and increasing function $\phi: [1, \infty) \to (0, \infty)$ defined by
\begin{equs}
\phi(s) = \frac{As}{ 2+\log s}
\end{equs}
for some $A>0$. We will show that if $A$ is large enough,
\begin{equs}[eq:ineqLF]
	LF \leq \phi(F)~,
\end{equs}
which implies the desired result. Indeed, assume that \eref{eq:ineqLF} holds. Let $(C_n)_{n\geq 0}$ be an increasing sequence of compact sets such that $C_n\uparrow\Omega$, and consider the corresponding first exit times $\tau_n = \inf\{ t\geq 0: x_t \notin C_n\}$. We have $\tau_n \to \infty$ almost surely, since the process is non-explosive. By Dynkin's formula, we find
\begin{equs}
\mEE_x F(x_{t\wedge \tau_n}) - F(x)& = \mEE_x \int_0^{t\wedge \tau_n} 	LF(x_s) \dd s \leq \mEE_x \int_0^{t\wedge \tau_n} 	\phi(F(x_s)) \dd s \\
& \leq \mEE_x \int_0^{t} 	\phi(F(x_{s \wedge \tau_n})) \dd s~ \leq \int_0^{t} 	\phi(\mEE_x  F(x_{s \wedge \tau_n})) \dd s~,
\end{equs}
where the last inequality comes from Fubini's theorem and Jensen's inequality (since $\phi$ is concave). In other words, $g(t) \equiv \mEE_x F(x_{t\wedge \tau_n})$ satisfies the integral inequality $g(t) \leq g(0) + \int_0^t \phi(g(s)) \dd s $.
The solution of the ordinary differential equation $y'(t) = \phi(y(t))$ with $y(0) = y_0 \geq 1$ is 
\begin{equ}[eq:solequadiff]
y(t) = \exp\big({\sqrt{(\ln(y_0)+2)^2 + 2At } }-2\big) \leq y_0 e^{ \sqrt{2At}}~,
\end{equ}
where we have used that $\sqrt{\argcdot}$ is subadditive. By comparison, we thus obtain that $\mEE_x F(x_{t\wedge \tau_n}) \leq F(x)  \exp(\sqrt {2At})$. Taking the limit $n\to \infty$ and using Fatou's lemma gives \eref{eq:borneEFxt}.

Thus, it only remains to prove \eref{eq:ineqLF}. Below, the constants $c_1, c_2, \dots$ are positive, and may depend on the parameters at hand, but {\em not} on the point in $\Omega$.

First, observe that there is a polynomial $z(p_1, p_2)$ such that
\begin{equs}[eq:Lrhoebetaplus]
L \big(\rho(p)e^{\frac{\beta_+}2 \pbarsq}\big)	&\leq c_1+ \ind_{\Omega_3}Le^{\frac{\beta_+}2 \pbarsq} + \ind_{\Omega_2}z(p)e^{\frac{\beta_+}2 \pbarsq}\\
& \leq  c_1 + \ind_{\Omega_3}c_2p_2^{-2}e^{\frac{\beta_+}2 \pbarsq} + \ind_{\Omega_2}z(p)e^{\frac{\beta_+}2 \pbarsq}~,
\end{equs}
where we have used \eref{eq:lomegasscndlower}.

Moreover, since $\beta_- < 1/T$,
\begin{equs}
Le^{\beta_- H}  &= \left((\beta_- T-1) p_1^{2} +
 T\right)\gamma \beta_-e^{\beta_- H} \leq (c_3- c_4p_1^{2})e^{\beta_- H}~.
\end{equs}
Introducing the set
\begin{equs}
G= \{x\in \Omega: p_1^2 < (1+c_3)/c_4\}	
\end{equs}
leads to
\begin{equs}[eq:LebetamoinsH]
Le^{\beta_- H}  & \leq c_3\ind_G e^{\beta_- H} - e^{\beta_- H} \leq c_5\ind_G e^{\beta_- \frac{p_2^2}2} - e^{\beta_- H}~.
\end{equs}

Combining \eref{eq:Lrhoebetaplus} and \eref{eq:LebetamoinsH}, we find
\begin{equs}[eq:LVleqfirstlower]
L F	&\leq  c_1 + \ind_{\Omega_3}c_2p_2^{-2}e^{\frac{\beta_+}2 \pbarsq} + \ind_{\Omega_2}z(p)e^{\frac{\beta_+}2 \pbarsq} + c_5\ind_G e^{\beta_- \frac{p_2^2}2} - e^{\beta_- H}~.
\end{equs}

We now make two observations. First, on $\Omega_2$, we have by \eref{eq:lomegasfirstlower}, the definition of $H$, and the definition of $\Omega_2$ that
\begin{equs}
z(p)e^{\frac{\beta_+}2 \pbarsq}e^{-\beta_- H}  \leq c_6z(p) e^{\frac{\beta_+}2 p_2^2 - \frac{\beta_-}{2}(p_1^2 + p_2^2)} \leq c_6 z(p) e^{\frac {p_2^2}{2}\left(\beta_+ - \beta_-\left(1+\frac 1{(1+2\delta)^2}\right)\right)}~.
\end{equs}
By \eref{eq:condbetaspm}, the above goes to zero when $\|p\|\to \infty$ in $\Omega_2$, so we have 
\begin{equ}[eq:indOmega2zedleq]
	\ind_{\Omega_2}z(p)e^{\frac{\beta_+}2 \pbarsq} - e^{\beta_- H} \leq c_7~.
\end{equ}

In a similar way, since $\beta_+ > \beta_-$ and $G \subset \Omega_3 \cup K$ for some compact set $K$ (on which $\exp({\beta_- \frac{p_2^2}2})$ is bounded), we have
\begin{equs}[eq:indGbetaminusleq]
\ind_G e^{\beta_- \frac{p_2^2}2} \leq c_8 + \ind_{\Omega_3}c_9p_2^{-2}e^{\frac{\beta_+}2 \pbarsq}~,
\end{equs}
where we have also used \eref{eq:lomegasfirstlower}.
Combining now \eref{eq:LVleqfirstlower}, \eref{eq:indOmega2zedleq} and \eref{eq:indGbetaminusleq}, we obtain
\begin{equs}
L F	&\leq  c_{10} + \ind_{\Omega_3}c_{11}p_2^{-2}e^{\frac{\beta_+}2 \pbarsq} \leq c_{10} + \ind_{\Omega_3}  \frac{c_{12} e^{\frac{\beta_+}2 \pbarsq}}{2+\log\left(e^{\frac{\beta_+}2 \pbarsq}\right)} ~,
\end{equs}
where the second inequality uses once more \eref{eq:lomegasfirstlower}. Since $\ind_{\Omega_3} \exp({\frac{\beta_+}2 \pbarsq})\leq F$, and since the function $s \mapsto  s/(2+\log s)$ is increasing, we obtain
\begin{equs}
L F	&\leq c_{10} + \frac{c_{12}F}{2+\log F} \leq \frac{c_{13}F}{2+\log F}~,
\end{equs}
where the second inequality holds because $F\geq 1$. Thus, we indeed have \eref{eq:ineqLF} for large enough $A$, which completes the proof.
\end{proof}

\begin{remark}\label{rem:lowerbound2seulement} 
The fact that we have to work with two different constants $\beta_+$ and $\beta_-$ seems to force us to take a ``transition region'' $\Omega_2$ where $p_1$ scales linearly with $p_2$ (all other attempts have resulted in some troublesome terms coming from the cutoffs). Unlike in \cite{CEP_nonequilibrium_2014} and §2 of \cite{four_rotors_2015}, we are therefore not allowed to assume that $p_1$ is small when we compute $\bar p_2$, which forces us to work with negative powers of $(p_2-p_1)$ instead of simply $p_2$. While this causes no trouble here with only two rotors, technical complications arise if we try to generalize the computations above to chains of three or four rotors. For example, terms involving $(p_2-p_1)^{-j}$ and $(p_2-p_3)^{-j}$ have to be combined, and this leads to troublesome error terms. In addition, with three or four rotors and two heat baths at different temperatures, the invariant measure $\pi$ is not known explicitly, and some supplementary work would have to be done to prove that the function $F$ satisfies $\pi(F^{1-\varepsilon}) = \infty$ for some $\varepsilon > 0$. 
\end{remark}

\subsection*{Acknowledgments} This work has been partially supported by the
ERC Advanced Grant ``Bridges'' 290843. The authors would like to thank Jean-Pierre Eckmann for helpful discussions and suggestions. N.~C.~also thanks the University of Geneva, where he was when the results in this paper were obtained.
\bibliography{refs}

\begin{thebibliography}{10}

\bibitem{carmona_2007}
P.~Carmona.
\newblock Existence and uniqueness of an invariant measure for a chain of
  oscillators in contact with two heat baths.
\newblock {\em Stochastic Process. Appl.\/} {\bf 117} (2007), 1076--1092.

\bibitem{cooke2011geometric}
B.~Cooke, D.~P. Herzog, J.~C. Mattingly, S.~A. McKinley, and S.~C. Schmidler.
\newblock Geometric ergodicity of two--dimensional hamiltonian systems with a
  lennard--jones--like repulsive potential.
\newblock {\em {arXiv} preprint\/}  (2011), arXiv:1104.3842.

\bibitem{four_rotors_2015}
N.~Cuneo and J.-P. Eckmann.
\newblock Non-equilibrium steady states for chains of four rotors.
\newblock {\em Commun. Math. Phys.\/} {\bf 345} (2016), 185--221.

\bibitem{CEP_nonequilibrium_2014}
N.~Cuneo, J.-P. Eckmann, and C.~Poquet.
\newblock Non-equilibrium steady state and subgeometric ergodicity for a chain
  of three coupled rotors.
\newblock {\em Nonlinearity\/} {\bf 28} (2015), 2397–2421.

\bibitem{CEW2017}
N.~Cuneo, J.-P. Eckmann, and C.~E. Wayne.
\newblock Energy dissipation in {H}amiltonian chains of rotators.
\newblock {\em {arXiv} preprint\/}  (2017), arXiv:1702.06464.

\bibitem{eckmann_hairer_2000}
J.-P. Eckmann and M.~Hairer.
\newblock Non-equilibrium statistical mechanics of strongly anharmonic chains
  of oscillators.
\newblock {\em Commun. Math. Phys.\/} {\bf 212} (2000), 105--164.

\bibitem{eckmann_nonequilibrium_1999}
J.-P. Eckmann, C.-A. Pillet, and L.~Rey-Bellet.
\newblock Non-equilibrium statistical mechanics of anharmonic chains coupled to
  two heat baths at different temperatures.
\newblock {\em Commun. Math. Phys.\/} {\bf 201} (1999), 657--697.

\bibitem{hairer_how_2009}
M.~Hairer.
\newblock How hot can a heat bath get?
\newblock {\em Commun. Math. Phys.\/} {\bf 292} (2009), 131--177.

\bibitem{hairer_slow_2009}
M.~Hairer and J.~C. Mattingly.
\newblock Slow energy dissipation in anharmonic oscillator chains.
\newblock {\em Comm. Pure Appl. Math.\/} {\bf 62} (2009), 999--1032.

\bibitem{li2016polynomial2}
Y.~Li.
\newblock On the polynomial convergence rate to nonequilibrium steady-states.
\newblock {\em {arXiv} preprint\/}  (2016), arXiv:1607.08492.

\bibitem{li2016polynomial}
Y.~Li and L.-S. Young.
\newblock Polynomial convergence to equilibrium for a system of interacting
  particles.
\newblock {\em {arXiv} preprint\/}  (2016), arXiv:1601.00717.

\bibitem{reybellet_exponential_2002}
L.~Rey-Bellet and L.~E. Thomas.
\newblock Exponential convergence to non-equilibrium stationary states in
  classical statistical mechanics.
\newblock {\em Commun. Math. Phys.\/} {\bf 225} (2002), 305--329.

\bibitem{MR3071443}
T.~Yarmola.
\newblock Sub-exponential mixing of random billiards driven by thermostats.
\newblock {\em Nonlinearity\/} {\bf 26} (2013), 1825--1837.

\bibitem{MR3217533}
T.~Yarmola.
\newblock Sub-exponential mixing of open systems with particle-disk
  interactions.
\newblock {\em J. Stat. Phys.\/} {\bf 156} (2014), 473--492.

\end{thebibliography}

\end{document}